\documentclass[sigconf]{acmart}

\usepackage{booktabs} 
\usepackage{amsmath}
\usepackage[ruled]{algorithm2e}
\usepackage{pgfplots}
\usetikzlibrary{positioning}
\usepackage{tikz}
\usetikzlibrary{backgrounds}
\tikzset{>=latex}
\usetikzlibrary{arrows,shapes}
\usetikzlibrary{positioning,fit,calc}
\usepackage{caption}
\usepackage{subcaption}
\usepackage{amssymb}
\usepackage{gensymb}
\usepackage{graphicx}
\usepackage{multirow}

\setcopyright{rightsretained}

\begin{document}
\title{SPIN: A Fast and Scalable Matrix Inversion Method in Apache Spark}

\author{Chandan Misra}
\affiliation{%
  \institution{Indian Institute of Technology Kharagpur}
  \state{West Bengal, India} 
}
\email{chandan.misra1@gmail.com}

\author{Swastik Haldar}
\affiliation{%
	\institution{Indian Institute of Technology Kharagpur}
	\state{West Bengal, India} 
}
\email{swastik.haldar@iitkgp.ac.in}

\author{Sourangshu Bhattacharya}
\affiliation{%
  \institution{Indian Institute of Technology Kharagpur}
  \state{West Bengal, India} 
}
\email{sourangshu@cse.iitkgp.ernet.in}

\author{Soumya K. Ghosh}
\affiliation{%
  \institution{Indian Institute of Technology Kharagpur}
  \state{West Bengal, India}}
\email{skg@iitkgp.ac.in}

\renewcommand{\shortauthors}{C. Misra et al.}

\begin{abstract}
The growth of big data in domains such as Earth Sciences, Social Networks, Physical Sciences, etc. has lead to an immense need for efficient and scalable linear algebra operations, e.g. Matrix inversion. Existing methods for efficient and distributed matrix inversion using big data platforms rely on LU decomposition based block-recursive algorithms. However, these algorithms are complex and require a lot of side calculations, e.g. matrix multiplication, at various levels of recursion. In this paper, we propose a different scheme based on Strassen's matrix inversion algorithm (mentioned in Strassen's original paper in 1969), which uses far fewer operations at each level of recursion. We implement the proposed algorithm, and through extensive experimentation, show that it is more efficient than the state of the art methods. Furthermore, we provide a detailed theoretical analysis of the proposed algorithm, and derive theoretical running times which match closely with the empirically observed wall clock running times, thus explaining the U-shaped behaviour w.r.t. block-sizes.
\end{abstract}

%
%


\copyrightyear{2018} 
\acmYear{2018} 
\setcopyright{acmcopyright}
\acmConference[ICDCN '18]{19th International Conference on Distributed Computing and Networking}{January 4--7, 2018}{Varanasi, India}
\acmPrice{15.00}
\acmDOI{10.1145/3154273.3154300}
\acmISBN{978-1-4503-6372-3/18/01}



\maketitle

\section{Introduction}
Dense matrix inversion is a basic procedure used by many applications in Data Science, Earth Science, Scientific Computing, etc, and has become an essential component of many such systems. It is an expensive operation, both in terms of computational and space complexity, and hence consumes a large fraction of resources in many of the workloads.
In the big data era, many of these applications have to work on huge matrices, possibly stored over multiple servers, and thus consuming huge amounts of computational resources for matrix inversion. Hence, designing efficient large scale distributed matrix inversion algorithms, is an important challenge.

Since its release in 2012, \textit{Spark} \cite{zaharia2010spark} has been adopted as a dominant solution for scalable and fault-tolerant processing of huge datasets in many applications, e.g., machine learning \cite{meng2016mllib}, graph processing \cite{gonzalez2014graphx}, climate science \cite{palamuttam2015scispark}, social media analytics \cite{al2017mapreduce}, etc. Spark has gained its popularity for its in-memory distributed data processing ability, which runs interactive and iterative applications faster than Hadoop MapReduce. It's close intergration with Scala / Java, and the flexible structure for RDDs allow distributed recursive algorithms to be implemented efficiently, without compromising on scalability and fault-tolerance.
Hence, in this paper we focus on Spark for implementation of large scale distributed matrix inversion.

There are a variety of existing inversion algorithms, e.g. methods based on QR decomposition \cite{press2007numerical}, LU decomposition \cite{press2007numerical}, Cholesky decomposition \cite{burian2003fixed}, Gaussian Elimination \cite{althoen1987gauss}, etc. Most of them require $O(n^{3})$ time (where $n$ denotes the order of the matrix), and main speed-ups in shared memory settings come from architecture specific optimizations (reviewed in section \ref{sec:related-work}). Surprisingly, there are not many studies on distributed matrix inversion using big-data frameworks, where jobs could be distributed over machines with a diverse set of architectures.
LU decomposition is the most widely used technique for distributed matrix inversion, possibly due to it's efficient block-recursive structure. Xiang et al. \cite{xiang2014scalable} proposed a Hadoop based implementation of inverting a matrix relying on computing the LU decomposition and discussed many Hadoop specific optimizations. Recently, Liu et al. \cite{liu2016spark} proposed several optimized block-recursive inversion algorithms on Spark based on LU decomposition. 
In the block recursive approach \cite{liu2016spark}, the computation is broken down into subtasks that are computed as a pipeline of Spark tasks on a cluster. The costliest part of the computation is the matrix multiplication and the authors have given a couple of optimized algorithms to reduce the number of multiplications. However, in spite of being optimized, the implementation requires $9$ $O(n^{3})$ operations on the leaf node of the recursion tree, $12$ multiplications at each recursion level of LU decomposition and an additional $7$ multiplication after the LU decomposition to invert the matrix, which makes the implementation perform slower. 

In this paper, we use a much simpler and less exploited algorithm, proposed by Strassen in his 1969 multiplication paper \cite{strassen1969gaussian}. The algorithm follows similar block-recursion structure as LU decompostion, yet providing a simpler approach to matrix inversion. This approach involves no additional matrix multiplication at the leaf level of recursion, and requires only $6$ multiplications at intermediate levels. We propose and implement a distributed matrix inversion algorithm based on Strassen's original serial inversion scheme. We also provide a detailed analysis of wall clock time for the proposed algorithm, thus revealing the `U'-shaped behaviour with respect to block size. Experimentally, we show comprehensively, that the proposed approach is superior to the LU decomposition based approaches for all corresponding block sizes, and hence overall. We also demonstrate that our analysis of the proposed approach matches with the empirically observed wall clock time, and similar to ideal scaling behaviour.
In summary:
\begin{enumerate}
	\item We propose and implement a novel approach (SPIN) to distributed matrix inversion, based on an algorithm proposed by Strassen \cite{strassen1969gaussian}.
	\item We provide a theoretical analysis of our proposed algorithm which matches closely with the empirically observed wall clock time.
	\item Through extensive experimentation, we show that the proposed algorithm is superior to the LU decomposition based approach.
\end{enumerate}

\section{Related Work}
\label{sec:related-work}
The literature on parallel and distributed matrix inversion can be divided broadly into three categories: 1) HPC based approach, 2) GPU based approach and 3) Hadoop and Spark based approach. Here, we briefly review them.

\subsection{HPC based approach}
LINPACK, LAPACK and ScaLAPACK are some of the most robust linear algebra software packages that support matrix inversion. LINPACK was written in Fortran and used on shared-memory vector computers. It has been superseded by LAPACK which runs more efficiently on modern cache-based architectures. LAPACK has also been extended to run on distributed-memory MIMD parallel computers in ScaLAPACK package. However, these packages are based on architectures and frameworks which are not fault tolerant and MapReduce based matrix inversion are more scalable than ScaLAPACK as shown in \cite{xiang2014scalable}. Lau et al. \cite{lau1996parallel} presented two algorithms for inverting sparse, symmetric and positive definite matrices on SIMD and MIMD respectively. The algorithm uses Gaussian elimination technique and the sparseness of the matrix to achieve higher performance. Bientinesi et al. \cite{bientinesi2008families} presented a parallel implementation of symmetric positive definite matrix on three architechtures --- sequential processors, symmetric multi-processors and distributed memory parallel computers using Cholesky factorization technique. Yang et al. \cite{yang2013parallel} presented a parallel algorithm for matrix inversion based on Gauss-Jordan elimination with partial pivoting. It used efficient mechanism to reduce the communication overhead and also provides good scalability. Bailey et al. presented techniques to compute inverse of a matrix using an algorithm suggested by Strassen in \cite{strassen1969gaussian}. It uses Newton iteration method to increase its stability while preserving parallelism. Most of the above works are based on specialized matrices and not meant for general matrices. In this paper, we concentrate on any kind of square positive definite and  invertible matrices which are distributed on large clusters which the above algorithms are not suitable for.

\subsection{Multicore and GPU based approach}
In order to fully exploit the multicore architecture, tile algorithms have been developed. Agullo et al. \cite{agullo2010towards} developed such a tile algorithm to invert a symmetric positive definite matrix using Cholesky decomposition.
Sharma et al. \cite{sharma2013fast} presented a modified Gauss-Jordan algorithm for matrix inversion on CUDA based GPU platform and studied the performance metrics of the algorithm. Ezzatti et al. \cite{ezzatti2011high} presented several algorithms for computing matrix inverse based on Gauss-Jordan algorithm on hybrid platform consisting of multicore processors connected to several GPUs. Although the above works have demonstrated that GPU can considerably reduce the computational time of matrix inversion, they are non-scalable centralized methods and need special hardwares.

\subsection{MapReduce based approach}
MadLINQ \cite{qian2012madlinq} offered a highly scalable, efficient and fault tolerant matrix computation system with a unified programming model which integrates with DryadLINQ, data parallel computing system. However, it does not mention any inversion algorithm explicitly. Xiang et al. \cite{xiang2014scalable} implemented first LU decomposition based matrix inversion in Hadoop MapReduce framework. However, it lacks typical Hadoop shortcomings like redundant data communication between map and reduce phases and inability to preserve distributed recursion structure. Liu et al. \cite{liu2016spark} provides the same LU based distributed inversion on Spark platform. It optimizes the algorithm by eliminating redundant matrix multiplications to achieve faster execution. Almost all the MapReduce based approaches relies on LU decomposition to invert a matrix. The reason is that it partitions the computation in a way suitable for MapReduce based systems. In this paper, we show that matrix inversion can be performed efficiently in a distributed environment like Spark by implementing Strassen's scheme which requires less number of multiplications than the earlier providing faster execution.

\section{Algorithm Design}
In this section, we discuss the implementation of \textit{SPIN} on Spark framework. First, we describe the original Strassen's inversion algorithm \cite{strassen1969gaussian} for serial matrix inversion in section \ref{sec:serial-strassen}. Next, in section \ref{sec:block-matrix-data-structure}, we describe the \textit{BlockMatrix} data structure from \textit{MLLib} which is used in our algorithm to distribute the large input matrix into the distributed file system. Finally, section \ref{sec:distributed-block-recursive} describes the distributed inversion algorithm, and its implementation strategy using \textit{Blockmatrix}.

\subsection{Strassen's Algorithm for Matrix Inversion}
\label{sec:serial-strassen}
Strassen's matrix inversion algorithm appeared in the same paper in which the well known Strassen's matrix multiplication was published. This algorithm can be described as follows. Let two matrices $A$ and $C=A^{-1}$ be split into half-sized sub-matrices:

\begin{center}
	$\begin{bmatrix}
		A_{11} & A_{12}\\ 
		A_{21} & A_{22}
		\end{bmatrix}^{-1} = 
		\begin{bmatrix}
		C_{11} & C_{12}\\ 
		C_{21} & C_{22}
	\end{bmatrix}$
\end{center}

Then the result $C$ can be calculated as shown in Algorithm \ref{alg:serial-strassen}. Intuitively, the steps involved in the algorithm are difficult to be performed in parallel. However, for input matrices which are too large to be fit into the memory on a single server, each such step is required to be processed distributively. These steps include breaking a matrix into four equal size sub-matrices, multiplication and subtraction of two matrices, multiplying a matrix to a scalar and arranging four half-sized sub-matrices into a full matrix. All these steps are done by splitting the matrix into blocks which act as execution unit of the spark job. A brief description of the block data structure is given below.

\begin{algorithm}	
    \SetKwInOut{Input}{Input}
    \SetKwInOut{Output}{Output}
    
    function Inverse$()$\;
    	\Input{Matrix $A$ (input matrix of size $n \times n$), int $threshold$}
    	\Output{Matrix $C$ (invert of matrix $A$}
    	\Begin{
        	\eIf{n=threshold}{
				invert A in any approach (e.g., LU, QR, SVD decomposition)\;
			}
            {
            	Compute $A_{11}, B_{11}, ..., A_{22}, B_{22}$ by computing $n=\frac{n}{2}$\;
    			$I \leftarrow A_{11}^{-1}$ \\
    			$II \leftarrow A_{21}.I$ \\
    			$III \leftarrow I.A_{12}$ \\
    			$IV \leftarrow A_{21}.III$ \\
    			$V \leftarrow IV-A_{22}$ \\
    			$VI \leftarrow V^{-1}$ \\
    			$C_{12} \leftarrow III.VI$ \\
    			$C_{21} \leftarrow VI.II$ \\
    			$VII \leftarrow III.C_{21}$ \\
    			$C_{11} \leftarrow I-VII$ \\
    			$C_{22} \leftarrow -VI$ \\
            }
    		
    		\Return $C$
    	}
    \caption{Strassen's Serial Inversion Algorithm}
	\label{alg:serial-strassen}
\end{algorithm}

\subsection{Block Matrix Data Structure}
\label{sec:block-matrix-data-structure}
In order to distribute the matrix in the HDFS (Hadoop Distributed File System), we create a distributed matrix called BlockMatrix, which is basically an RDD of MatrixBlocks spread in the cluster. Distributing the matrix as a collection of Blocks makes them easy to be processed in parallel and follow divide and conquer approach.\textit{MatrixBlock} is a block of matrix represented as a tuple \textit{((rowIndex, columnIndex), Matrix)}. Here, \textit{rowIndex} and \textit{columnIndex} are the row and column index of a block of the matrix. \textit{Matrix} refers to a one-dimensional array representing the elements of the matrix arranged in a column major fashion.

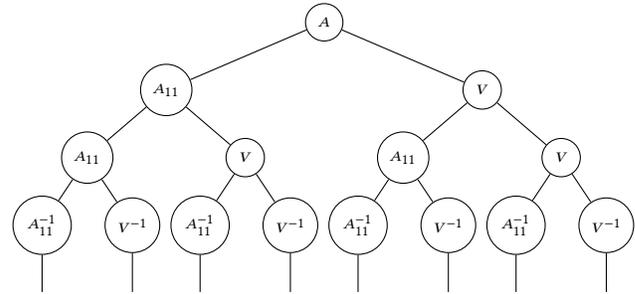
\begin{figure}
	\centering
	\begin{tikzpicture}[scale=0.6,level 1/.style={sibling distance=70mm},level 2/.style={sibling distance=35mm, style={draw}},level 3/.style={sibling distance=20mm},circleNodeEmpty/.style={rectangle}]
	\node [circle,draw] (z){\tiny $A$}
	child {node [circle,draw] (a) {\tiny $A_{11}$} 
		child {node [circle,draw] (b) {\tiny $A_{11}$}
			child {node  [circle,draw] (d) {\tiny $A_{11}^{-1}$}} 
			child {node  [circle,draw] (e) {\tiny $V^{-1}$}}
		}
		child {node [circle,draw] (g) {\tiny $V$}
			child {node  [circle,draw] (M) {\tiny $A_{11}^{-1}$}}
			child {node  [circle,draw] (N) {\tiny $V^{-1}$}}
		}
	}
	child {node [circle,draw] (j) {\tiny $V$}
		child {node [circle,draw] (k) {\tiny $A_{11}$}
			child {node  [circle,draw] (O) {\tiny $A_{11}^{-1}$}}
			child {node  [circle,draw] (P) {\tiny $V^{-1}$}}
		}
		child {node [circle,draw] (l) {\tiny $V$}
			child {node [circle,draw] (Q) {\tiny $A_{11}^{-1}$}}
			child {node [circle,draw] (o) {\tiny $V^{-1}$}}
		}
	};
    
    \node[circleNodeEmpty](1)[below of = d]{};
    \node[circleNodeEmpty](2)[below of = e]{};
    \node[circleNodeEmpty](3)[below of = M]{};
    \node[circleNodeEmpty](4)[below of = N]{};
    \node[circleNodeEmpty](5)[below of = O]{};
    \node[circleNodeEmpty](6)[below of = P]{};
    \node[circleNodeEmpty](7)[below of = Q]{};
    \node[circleNodeEmpty](8)[below of = o]{};
    
    \draw (d.south) -- (1.north);
    \draw (e.south) -- (2.north);
    \draw (M.south) -- (3.north);
    \draw (N.south) -- (4.north);
    \draw (O.south) -- (5.north);
    \draw (P.south) -- (6.north);
    \draw (Q.south) -- (7.north);
    \draw (o.south) -- (8.north);
	\end{tikzpicture}	
	\caption{Recursion tree for Algorithm \ref{alg:distributed-strassen}}
	\label{fig:recursion-tree}
\end{figure}

\subsection{Distributed Block-recursive Matrix Inversion Algorithm}
\label{sec:distributed-block-recursive}
The distributed block-recursive algorithm can be visualized as Figure \ref{fig:recursion-tree}, where upper left sub-matrix is divided recursively until it can be inverted serially on a single machine. After the leaf node inversion, the inverted matrix is used to compute intermediate matrices, where each step is done distributively. Another recursive call is performed for matrix $VI$ until leaf node is reached. Like $A_{11}$, it is also inverted on a single node when the leaf node is reached. The core inversion algorithm (described in Algorithm \ref{alg:distributed-strassen}) takes a matrix (say $A$) represented as \textit{BlockMatrix}, as input as shown in Figure \ref{fig:recursion-tree}. The core computation performed by the algorithm is based on six distributed methods, which are as follows:
\begin{itemize}
	\item \textit{breakMat}: Breaks a matrix into four equal sized sub-matrices
	\item \textit{xy}: Returns one of the four sub-matrices after the breaking,  according to the index specified by $x$ and $y$.
	\item \textit{multiply}: Multiplies two \textit{BlockMatrix}
	\item \textit{subtract}: Subtracts two \textit{BlockMatrix}
	\item \textit{scalarMul}: Multiples a scalar with a \textit{BlockMatrix}
	\item \textit{arrange}: Arranges four equal quarter \textit{BlockMatrices} into a single full \textit{BlockMatrix}.
\end{itemize}

Below we describe the methods in a little bit more details and also provide the algorithm for each.

\begin{algorithm}	
    \SetKwInOut{Input}{Input}
    \SetKwInOut{Output}{Output}

    function Inverse$()$\;
    \Begin{
    	\Input{BlockMatrix $A$, int $size$, int $blockSize$}
    	\Output{BlockMatrix $AInv$}
        $size =$ Size of matrix $A$ or $B$\;
    	$blockSize =$ Size of a single matrix block\;
    	$n =$ $\frac{size}{blockSize}$\;
    	\eIf{$n=1$}
    	{
    		$RDD<Block> invA\leftarrow A.toRDD()$ \\
    		Map()\;
    		\Begin{
    			\Input{Block block}
    			\Output{Block block}
    			$block.matrix\leftarrow locInverse(block.matrix)$ \\
    			\Return $block$
    		}
    		$blockAInv\leftarrow invA.toBlockMatrix()$ \\
    		\Return $blockAInv$
    	}
    	{
    		$size\leftarrow size/2$ \\
    		$pairRDD\leftarrow breakMat(A, size)$ \\
    		$A11\leftarrow \_11(pairRDD, blockSize)$ \\
    		$A12\leftarrow \_12(pairRDD, blockSize)$ \\
    		$A21\leftarrow \_21(pairRDD, blockSize)$ \\
    		$A22\leftarrow \_22(pairRDD, blockSize)$ \\
    		$I\leftarrow Inverse(A11, size, blockSize)$ \\
    		$II\leftarrow multiply(A21,I)$ \\
    		$III\leftarrow multiply(I,A12)$ \\
    		$IV\leftarrow multiply(A21,III)$ \\
    		$V\leftarrow subtract(IV,A22)$ \\
    		$VI\leftarrow Inverse(V,size, blockSize)$ \\
    		$C12\leftarrow multiply(III,VI)$ \\
    		$C21\leftarrow multiply(VI,II)$ \\
    		$VII\leftarrow multiply(III,C21)$ \\
    		$C11\leftarrow subtract(I,VII)$ \\
    		$C22\leftarrow scalerMul(VI, -1, blockSize)$ \\
    		$C\leftarrow arrange(C11, C12, C21, C22, size, blockSize)$ \\
    		\Return $C$
    	}
}    
\caption{Spark Algorithm for Strassen's Inversion Scheme}
\label{alg:distributed-strassen}
\end{algorithm}

\textbf{\textit{breakMat}} method breaks a matrix into four sub-matrices but does not return four sub-matrices to the caller. It just prepare the input matrix to a form which help filtering each part easily. As described in Algorithm \ref{alg:breakMat} it takes a \textit{BlockMatrix} and returns a \textit{PairRDD} of \textit{tag} and \textit{Block} using a \textit{mapToPair} transformation. First, the \textit{BlockMatrix} is converted into an RDD of \textit{MatrixBlocks}. Then, each \textit{MatrixBlock} of the RDD is mapped to tuple of \textit{(tag, MatrixBlock)}, resulting a \textit{pairRDD} of such tuples. Inside the \textit{mapToPair} transformation, we carefully tag each \textit{MatrixBlock} according to which quadrant it belongs to.

\begin{algorithm}
	\SetKwInOut{Input}{Input}
	\SetKwInOut{Output}{Output}
	
	function breakMat$()$\;
	\Begin{
		\Input{BlockMatrix A, int size}
		\Output{PairRDD brokenRDD}		
		$ARDD\leftarrow A.toRDD$ \\
		MapToPair()\;
		\Begin{
			\Input{block of ARDD}
			\Output{tuple of brokenRDD}
			$ri\leftarrow block.rowIndex$ \\
			$ci\leftarrow block.colIndex$ \\
			\If{$ri/size=0$ \& $ci/size=0$}
			{
				$tag\leftarrow ``A11"$ \\
			}
			\ElseIf{$ri/size=0$ \& $ci/size=1$}
			{
				$tag\leftarrow ``A12"$ \\
			}
			\ElseIf{$ri/size=1$ \& $ci/size=0$}
			{
				$tag\leftarrow ``A21"$ \\
			}
			\Else{
				$tag\leftarrow ``A22"$ \\
			}
			$block.rowIndex\leftarrow ri\%size$ \\
			$block.colIndex\leftarrow ci\%size$ \\
			\Return $Tuple2(tag, block)$
		}
		
		\Return $brokenMat$
	}    
	\caption{Spark Algorithm for breaking a \textit{BlockMatrix}}
	\label{alg:breakMat}
\end{algorithm}

\textbf{\textit{xy}} method is a generic method signature for four methods used for accessing one of the four sub-matrices of size $2^{n-1}$ from a matrix of size $2^{n}$. Each method consists of two transformation --- \textit{filter} and \textit{map}. \textit{filter} takes the matrix as a \textit{pairRDD} of \textit{(tag, MatrixBlock)} tuple which was the output of \textit{breakMat} method and filters the appropriate portion against the tag associated with the \textit{MatrixBlock}. Then it converts the \textit{pairRDD} into \textit{RDD} using the \textit{map} transformation.

\begin{algorithm}
	\SetKwInOut{Input}{Input}
	\SetKwInOut{Output}{Output}
	
	function xy$()$\;
	\Begin{
		\Input{PairRDD $brokenRDD$}
		\Output{BlockMatrix $xy$}
        filter()\;
        \Begin{
        	\Input{PairRDD $brokenRDD$}
			\Output{PairRDD $filteredRDD$}
            \Return $brokenRDD.tag = ``A_{xy}"$
        }
        map()\;
        \Begin{
        	\Input{PairRDD $filteredRDD$}
			\Output{RDD $rdd$}
            \Return $filteredRDD.block$
        }
        $xy\leftarrow rdd.toBlockMatrix()$ \\
        \Return $xy$
	}    
	\caption{Spark Algorithm for multiplying a scalar to a distributed matrix}
\end{algorithm}

\textbf{\textit{multiply}} method multiplies two input sub-matrices and returns another sub-matrix of \textit{BlockMatrix} type. Multiply method in our algorithm uses naive block matrix multiplication approach, which replicates the blocks of matrices and groups the blocks together to be multiplied in the same node. It uses co-group to reduce the communication cost.

\textbf{\textit{subtract}} method subtracts two \textit{BlockMatrix} and returns the result as \textit{BlockMatrix}.

\textbf{\textit{scalarMul}} method (as described in Algorithm \ref{alg:scalarMul}), takes a \textit{BlockMatrix} and returns another \textit{BlockMatrix} using a \textit{map} transformation. The \textit{map} takes blocks one by one and multiply each element of the block with the scalar.

\begin{algorithm}
	\SetKwInOut{Input}{Input}
	\SetKwInOut{Output}{Output}
	
	function scalarMul$()$\;
	\Begin{
		\Input{BlockMatrix A, double scalar, int blockSize}
		\Output{BlockMatrix productMat}
		$ARDD\leftarrow A.toRDD()$ \\
		Map()\;
		\Begin{
			\Input{block of ARDD}
			\Output{block of productRDD}
			$product\leftarrow block.matrix.toDoubleMatrix$ \\			
			$block.matrix\leftarrow product.toMatrix$ \\
			\Return $block$
		}
		$productMat\leftarrow product.toBlockMatrix()$ \\
		\Return $productMat$
	}    
	\caption{Spark Algorithm for multiplying a scalar to a distributed matrix}
	\label{alg:scalarMul}
\end{algorithm}

\textbf{\textit{arrange}} method (as described in Algorithm \ref{alg:arrange}), takes four sub-matrices of size $2^{n-1}$ which represents four co-ordinates of a full matrix of size $2^{n}$ and arranges them in later and returns it as \textit{BlockMatrix}. It consists of four \textit{map}s, each one for a separate \textit{BlockMatrix}. Each \textit{map} maps the block index to a different block index that provides the final position of the block in the result matrix.

\begin{algorithm}
	\SetKwInOut{Input}{Input}
	\SetKwInOut{Output}{Output}
	
	function arrange$()$\;
	\Begin{
		\Input{BlockMatrix C11, BlockMatrix C12, BlockMatrix C21, BlockMatrix C22, int size, int blockSize}
		\Output{BlockMatrix arranged}
			$C11RDD\leftarrow C11.toRDD()$ \\
			$C12RDD\leftarrow C12.toRDD()$ \\
			$C21RDD\leftarrow C21.toRDD()$ \\
			$C22RDD\leftarrow C22.toRDD()$ \\
			Map()\;
			\Begin{
				\Input{block of C12RDD}
				\Output{block of C1}
				$block.colIndex\leftarrow block.colIndex + size$ \\
				\Return $block$
			}
			Map()\;
			\Begin{
				\Input{block of C21RDD}
				\Output{block of C2}
				$block.rowIndex\leftarrow block.rowIndex + size$ \\
				\Return $block$
			}
			Map()\;
			\Begin{
				\Input{block of C22}
				\Output{block of C3}
				$block.rowIndex\leftarrow block.rowIndex + size$ \\
				$block.colIndex\leftarrow block.colIndex + size$ \\
				\Return $block$
			}
			$unionRDD\leftarrow C11RDD.union(C1.union(C2.union(C3)))$ \\
			$C\leftarrow unionRDD.toBlockMatrix()$ \\
			\Return $C$
	}    
	\caption{Spark Algorithm for rearranging four sub-matrices into single matrix}
	\label{alg:arrange}
\end{algorithm}

\section{Performance Analysis}
In this section, we attempt to estimate the performances of the proposed approach, and state-of-the-art approach using LU decomposition for distributed matrix inversion. In this work, we are interested in the \textit{wall clock running time} of the algorithms for varying number of nodes, matrix sizes and other algorithmic parameters e.g., partition / block sizes. This is because we are interested in the practical efficiency of our algorithm which includes not only the time spent by the processes in the CPU, but also the time taken while waiting for the CPU as well as data communication during shuffle. The wall clock time depends on three independently analyzed quantities: total \textit{computational complexity} of the sub-tasks to be executed, total \textit{communication complexity} between executors of different sub-tasks on each of the nodes, and \textit{parallelization factor} of each of the sub-tasks or the total number of processor cores available.

Later, in section \ref{sec:experiments}, we compare the theoretically derived estimates of \textit{wall clock time} with empirically observed ones, for validation. We consider only square matrices of dimension $2^{p}$ for all the derivations. The key input and tunable parameters for the algorithms are:

\begin{itemize}
\item $n = 2^p$: number of rows or columns in matrix $A$
\item $b$ = number of splits for square matrix
    \item $2^{q} = \frac{n}{b}$ = block size in matrix $A$
    \item $cores$ = Total number of physical cores in the cluster
    \item $i$ = current processing level of algorithm in the recursion tree.
    \item $m$ = total number of levels of the recursion tree.
\end{itemize}
Therefore,
\begin{itemize}
    \item Total number of blocks in matrix $A$ or $B$ = $b^{2}$
    \item $b$ = $2^{p-q}$
\end{itemize}

\begin{lemma}
The proposed distributed block recursive strassen's matrix inversion algorithm or \textit{SPIN} (presented in Algorithm \ref{alg:distributed-strassen}) has a complexity in terms of wall clock execution time requirement, where $n$ is the matrix dimension, $b$ is the number of splits, and $cores$ is the actual number of physical cores available in the cluster, as

\begin{equation}
\label{eq:SPIN-cost}
\begin{aligned}
Cost_{SPIN}&=\left(\frac{n^{3}}{b^{2}} \right )+
\frac{10b^{2}-6b}{min\left[\frac{b^{2}}{4^{i}},cores \right ]}
+\frac{(b-1)+\left(9b^{2}+n^{2}\left(b+1\right)\right)}{b\times min\left[\frac{b^{2}}{4^{i+1}},cores \right ]} \\
&+\frac{n^{2}(b^{2}n+b^{2}-2n)}{b^{2}\times min\left[\frac{n^{2}}{4^{i+1}},cores \right ]}
\end{aligned}
\end{equation}
\end{lemma}

\begin{proof}
Before going into details of the analysis, we give the performance analysis of the methods described in section \ref{sec:distributed-block-recursive}. A summary of the independently analyzed quantities is given in Table \ref{tab:cost-LU-Spin}.

\begin{table*}
\caption{Summary of the cost analysis of \textit{LU } and \textit{SPIN}}
\label{tab:cost-LU-Spin}
	\begin{center}
		\begin{tabular}{|c|c|c|c|c|}
           	\hline
           	\multirow{2}{*}{Method} & \multicolumn{2}{|c|}{Computation Cost} & \multicolumn{2}{c|}{Parallelization Factor} \\ \cline{2-5}
           	& LU & SPIN & LU & SPIN \\
           	\hline
           	leafNode & $9\times \frac{n^{3}}{b^{2}}$ & $\frac{n^{3}}{b^{2}}$ & --- & --- \\
           	\hline
           	breakMat & $\frac{2}{3}\left(b^{2}-3b+2 \right )$ & $2b^{2}-2b$ & $min\left[\frac{b^{2}}{4^{i}},cores \right ]$ &  $min\left[\frac{b^{2}}{4^{i}},cores \right ]$ \\
           	\hline
           	xy (filter) & $\frac{2}{3}\left(b^{2}-3b+2 \right )$ & $8b^{2}-4b$ & $min\left[\frac{b^{2}}{4^{i+1}},cores \right ]$ & $min\left[\frac{b^{2}}{4^{i}},cores \right ]$ \\
           	\hline
           	xy (map) & $\frac{1}{6}\left(b^{2}-3b+2 \right )$ & $2b^{2}-2b$ & $min\left[\frac{b^{2}}{4^{i+2}},cores \right ]$ & $min\left[\frac{b^{2}}{4^{i+1}},cores \right ]$ \\
           	\hline
           	multiply (large) & $\frac{16n^{3}}{21b^{3}}(b^{3}-7b+6)$ & $\frac{n^{3}}{6b^{2}}(b^{2}-1)$ & $min\left[\frac{n^{2}}{4^{i}},cores \right ]$ & $min\left[\frac{n^{2}}{4^{i+1}},cores \right ]$ \\
           	\hline
           	multiply Communication (large) & $\frac{8n^{2}(b^{2}-1)(8b^{2}-112)}{105b^{2}}$ & $\frac{n^{2}(b^{2}-1)}{6b}$ & $min\left[\frac{b^{2}}{4^{i}},cores \right ]$ & $min\left[\frac{b^{2}}{4^{i+1}},cores \right ]$\\
           	\hline
           	multiply (small) & $\frac{8n^{3}}{42b^{3}}(b^{3}-7b+6)$ & --- & $min\left[\frac{n^{2}}{4^{i+1}},cores \right ]$ & --- \\
           	\hline
           	multiply Communication (small) & $\frac{n^{2}(b^{2}-1)(8b^{2}-112)}{105b^{2}}$ & --- & $min\left[\frac{b^{2}}{4^{i+1}},cores \right ]$ & --- \\
           	\hline
           	subtract & $\frac{2n^{2}}{3b^{2}}(b^{2}-3b+2)$ & $\frac{n^{2}}{2b}(b-1)$ & $min\left[\frac{n^{2}}{4^{i}},cores \right ]$ & $min\left[\frac{n^{2}}{4^{i+1}},cores \right ]$\\
           	\hline
           	scalarMul & $\frac{4}{3}\left(b^{2}-3b+2 \right )$ & $\frac{b}{2}\left(b-1 \right )$ & $min\left[\frac{b^{2}}{4^{i}},cores \right ]$ & $min\left[\frac{b^{2}}{4^{i+1}},cores \right ]$ \\
           	\hline
           	arrange & --- & $\frac{b}{2}\left(b-1 \right )$ & --- & $min\left[\frac{b^{2}}{4^{i+1}},cores \right ]$ \\
           	\hline
           	Additional Cost & $7\times \left(\frac{n}{2} \right )^{3}$ & --- & $min\left[\frac{n^{2}}{4},cores \right ]$ & --- \\
           	\hline
		\end{tabular}
	\end{center}
\end{table*}

There are two primary part of the algorithm --- \textit{if} part and \textit{else} part. \textit{If} part does the calculation of the leaf nodes of the leaf nodes of the recursion tree as shown in Figure \ref{fig:recursion-tree}, while \textit{else} part does the computation for internal nodes. It is clearly seen from the figure that, at level $i$, there are $2^{i}$ nodes and the leaf level contains $2^{p-q}$ nodes.

There is only one transformation in \textit{if} part which is \textit{map}. It calculates the inverse of a matrix block in a single node using serial matrix inversion method. The size of each block is $n/b$ and we need $\approx(n/b)^{3}$ time to perform each such method. Therefore, the computation cost to process all the leaf nodes is 

\begin{equation}
\begin{aligned}
Comp_{leafNode}=2^{p-q}\times \left(\frac{n}{b}\right)^{3}=\frac{n^{3}}{b^{2}}
\end{aligned}
\end{equation}

In \textit{SPIN}, leaf nodes processes one block on a single machine of the cluster. In spite of being small enough to be accommodated in a single node, we do not collect them in the master node for the communication cost. Instead, we do a \textit{map} which takes the only block of the RDD, do the calculation and return the RDD again. 

\textbf{\textit{breakMat}} method takes a \textit{BlockMatrix} and returns a \textit{PairRDD} of \textit{tag} and \textit{Block} using a \textit{mapToPair} transformation. If the method is executed for $m$ levels, the computation cost of \textit{breakMat} is

\begin{equation}
\begin{aligned}
Comp_{breakMat}=\sum_{i=0}^{m-1}2^{i}\times \left(\frac{b^{2}}{4^{i}}\right)=2b\left(b-1\right)
\end{aligned}
\end{equation}

Note that, $i^{th}$ level contains $2^{i}$ nodes. Here each block is consumed in parallel giving parallelization factor as

\begin{equation}
\begin{aligned}
PF_{breakMat}&=min\left[\left(\frac{b^{2}}{4^{i}}\right),cores\right]
\end{aligned}
\end{equation}

The total number of blocks processed in \textit{filter} and \textit{map} are $\left(\frac{b^{2}}{4^{i}}\right)$ and $\left(\frac{b^{2}}{4^{i+1}}\right)$ for $i^{th}$ level respectively. Consequently, the parallelization factor of both of them are $min\left[\left(\frac{b^{2}}{4^{i}}\right),cores\right]$ and $min\left[\left(\frac{b^{2}}{4^{i+1}}\right),cores\right]$ respectively. Therefore, the computation cost for \textit{xy} is

\begin{equation}
\begin{aligned}
Comp_{xy}&=\left[\frac{\sum_{i=0}^{m-1}2^{i}\times \left(\frac{b^{2}}{4^{i}}\right)}{min\left[\left(\frac{b^{2}}{4^{i}}\right),cores\right]}+\frac{\sum_{i=0}^{m-1}2^{i}\times \left(\frac{b^{2}}{4^{i+1}}\right)}{min\left[\left(\frac{b^{2}}{4^{i+1}}\right),cores\right]}\right] \\
&=\left[\frac{8b^{2}-4b}{min\left[\left(\frac{b^{2}}{4^{i}}\right),cores\right]}+\frac{2b^{2}-2b}{min\left[\left(\frac{b^{2}}{4^{i+1}}\right),cores\right]}\right]
\end{aligned}
\end{equation}

\textbf{\textit{multiply}} method multiplies two \textit{BlockMatrices}, the computation cost of which can be derived as

\begin{equation}
\label{eq:multiply}
\begin{aligned}
Comp_{multiply}=\sum_{i=0}^{m-1}2^{i}\times \left(\frac{n^{3}}{8^{i+1}}\right)=\frac{n^{3}\left(b^{2}-1\right)}{6b^{2}}
\end{aligned}
\end{equation}

and the parallelization factor will be 

\begin{equation}
\begin{aligned}
PF_{multiply}&=min\left[\frac{n^{2}}{4^{i+1}},cores\right]
\end{aligned}
\end{equation}

\textbf{\textit{subtract}} method subtracts two \textit{BlockMatrices} using a \textit{map} transformation. There are two subtraction in each recursion level. Therefore,

\begin{equation}
\begin{aligned}
Comp_{subtract}=\sum_{i=0}^{m-1}2^{i}\times \left(\frac{n^{2}}{4^{i+1}}\right)=\frac{n^{2}\left(b-1 \right )}{2b}
\end{aligned}
\end{equation}

and the parallelization factor will be 

\begin{equation}
\begin{aligned}
PF_{subtract}&=min\left[\frac{n^{2}}{4^{i+1}},cores\right]
\end{aligned}
\end{equation}

\textbf{\textit{scalarMul}} method (as described in Algorithm \ref{alg:scalarMul}), takes a \textit{BlockMatrix} and returns another \textit{BlockMatrix} using a \textit{map} transformation. The \textit{map} takes blocks one by one and multiply each each element of the block with the scalar. Therefore, the computation cost of \textit{scalarMul} is

\begin{equation}
\label{eq:scalarMul}
\begin{aligned}
Comp_{scalarMul}=\sum_{i=0}^{m-1}2^{i}\times \left(\frac{b^{2}}{4^{i+1}}\right)=\frac{b}{2}\left(b-1\right)
\end{aligned}
\end{equation}

Again, here each block is consumed in parallel giving parallelization factor as

\begin{equation}
\begin{aligned}
PF_{scalarMul}&=min\left[\left(\frac{b^{2}}{4^{i+1}}\right),cores\right]
\end{aligned}
\end{equation}

\textbf{\textit{arrange}} method (as described in Algorithm \ref{alg:arrange}), takes four sub-matrices of size $2^{n-1}$ which represents four co-ordinates of a full matrix of size $2^{n}$ and arranges them in later and returns it as \textit{BlockMatrix}. It consists of four \textit{map}s, each one for a separate \textit{BlockMatrix}. Each \textit{map} maps the block index to a different block index that provides the final position of the block in the result matrix. The computation cost and parallelization factor for \textit{maps} are same as \textit{scalarMul}, which can be found in equation \ref{eq:scalarMul}.

%
%

\textit{SPIN} requires $4$ \textit{xy} method calls, $6$ multiplications, and $2$ subtractions for each recursion level. When summed up it will give equation \ref{eq:SPIN-cost}.

\end{proof}



\begin{lemma}
The proposed distributed block recursive LU decomposition based matrix inversion algorithm or \textit{SPIN} (presented in Algorithm 5, 6, and 7 in \citep{liu2016spark}) has a complexity in terms of wall clock execution time requirement, where $n$ is the matrix dimension, $b$ is the number of splits, and $cores$ is the actual number of physical cores available in the cluster, as below
\end{lemma}

\begin{equation}
\label{eq:LU-cost}
\begin{aligned}
Cost_{LU}&=\frac{9n^{3}}{b^{2}}+
\frac{(b-1)[210b^{2}(b-2)+64n^{2}(b+1)(b^{2}-14)]}{105b^{2}\times min\left[\frac{b^{2}}{4^{i}},cores \right ]} \\
&+\frac{(b-1)[70b^{2}(b-2)+8n^{2}(b+1)(b^{2}-14)]}{105b^{2}\times min\left[\frac{b^{2}}{4^{i+1}},cores \right ]} \\
&+\frac{(b-1)(b-2)}{105b^{2}\times min\left[\frac{b^{2}}{4^{i+2}},cores \right ]} \\
&+\frac{2n^{2}(b-1)[8n(b^{2}+b+6)+7b(b-2)]}{21b^{3}\times min\left[\frac{n^{2}}{4^{i}},cores \right ]} \\
&+\frac{8n^{3}(b-1)(b^{2}+b-6)}{42b^{3}\times min\left[\frac{n^{2}}{4^{i+1}},cores \right ]}
+ \frac{7n^{3}}{8\times min\left[\frac{n^{2}}{4},cores \right ]}
	\end{aligned}
\end{equation}

\begin{proof}
Liu et al. in \cite{liu2016spark} has described several algorithms for distributed matrix inversion using LU decomposition. We are referring the most optimized one (stated as Algorithm 5, 6 and 7 in the paper) for the performance analysis. The core computation of the algorithm is done with 1) a call to a recursive method \textit{LU} which basically decomposes the input matrix recursively until leaf nodes of the tree where the size of the matrix reaches the block size and 2) the computation after LU decomposition. The matrix inversion algorithm performs $7$ additional multiplications (as given in Algorithm 5 in \cite{liu2016spark}) of size $\left(\frac{n}{2}\right)$, providing additional cost of, which is basically $7$ matrix multiplications of dimension $\left(\frac{n}{2}\right)$. We call this as \textit{Additional Cost} and can be obtained as follows

\begin{equation}
\begin{aligned}
Comp_{AdditionalCost}= \frac{7\times \left(\frac{n}{2} \right )^{3}}{min\left[\frac{n^{2}}{4},cores \right ]}
\end{aligned}
\end{equation}

There are two primary parts of the \textit{LU} method --- \textit{if} part and \textit{else} part. \textit{If} part does the LU decomposition at the leaf nodes of the recursion tree while \textit{else} part does for the internal nodes. \textit{If} part requires $2$ LU decomposition, $4$ matrix inversion and $3$ matrix multiplications and there are $2^{p-q}$ number of leaf nodes in the recursion tree. Each of these processing requires $O(\left(\frac{n}{b}\right)^{3})$ time for a matrix of $n$ dimension. Therefore, the  total cost of the \textit{if} part is

\begin{equation}
\begin{aligned}
Comp_{leagNode}=9 \times 2^{p-q}\times (\frac{n}{b})^{3}=9\times \left(\frac{n^{3}}{b^{2}} \right )
\end{aligned}
\end{equation}

The \textit{else} part requires $4$ \textit{multiply}, $1$ \textit{subtraction} and $2$ calls to \textit{getLU} method. \textit{getLU} method compose the LU of a matrix by taking $9$ matrices of dimension $2^{k}$ and arranges them to return $3$ matrices of size $2^{k+1}$. It requires $4$ \textit{multiply} and $2$ \textit{scalarMul} methods of matrices of dimension $2^{k}$. 

The recursion scheme of LU decomposition is little bit different from \textit{SPIN}. Here the number of \textit{LU} call at level $i$ is $2^{i}-1$ instead of $2^{i}$ of \textit{SPIN}. The computation and communication costs for the methods (summarized in Table \ref{tab:cost-LU-Spin}) can be summed up to get equation \ref{eq:LU-cost}.
\end{proof}

\section{Experiments}
\label{sec:experiments}
In this section, we perform experiments to evaluate the execution efficiency of our implementation \textit{SPIN} comparing it with the distributed LU decomposition based inversion approach (to be mentioned as \textit{LU} from now) and scalability of the algorithm compared to ideal scalability. First, we select the fastest wall clock execution time among different partition size for each approach and compare them. Second, we conduct a series of experiments to individually evaluate the effect of partition size and matrix size of each competing approach. At last we evaluate the scalability of our implementation.

\subsection{Test Setup}
All the experiments are carried out on a dedicated cluster of 3 nodes. Software and hardware specifications are summarized in Table \ref{tab:test-setup}. Here \textit{NA} means \textit{Not Applicable}.

\begin{table}
	\caption{Summary of Test setup components specifications}
	\label{tab:test-setup}
	\begin{minipage}{\columnwidth}
		\begin{center}
			\begin{tabular}{lll}
				\toprule
				Component Name & Component Size & Specification \\
				\toprule
				Processor & 2 & Intel Xeon 2.60 GHz \\
				Core & 6 per processor & NA \\
				Physical Memory & 132 GB & NA \\
				Ethernet & 14 Gb/s & Infini Band \\
				OS & NA & CentOS 5 \\
				File System & NA & Ext3 \\
				Apache Spark & NA & 2.1.0 \\
				Apache Hadoop & NA & 2.6.0 \\
				Java & NA & 1.7.0 update 79 \\
				\bottomrule
			\end{tabular}
		\end{center}
	\end{minipage}
\end{table}

For block level multiplications both the implementation uses JBlas \cite{jblas}, a linear algebra library for Java based on \textit{BLAS} and \textit{LAPACK}. We have tested the algorithms on matrices with increasing cardinality from $(16 \times 16)$ to $(16384 \times 16384)$. All of these test matrices have been generated randomly using Java Random class. 

\paragraph{\textbf{Resource Utilization Plan}}
While running the jobs in the cluster, we customize three parameters --- the number of executors, the executor memory and the executor cores. We wanted a fair comparison among the competing approaches and therefore, we ensured jobs should not experience \textit{thrashing} and none of the cases tasks should fail and jobs had to be restarted. For this reason, we restricted ourselves to choose the parameters value which provides good utilization of cluster resources and mitigating the chance of task failures. By experimentation we found that, keeping executor memory as $50$ GB ensures successful execution of jobs without \textit{``out of memory''} error or any task failures for all the competing approaches. This includes the small amount of overhead to determine the full request to YARN for each executor which is equal to $3.5$ GB. Therefore, the executor memory is $46.5$ GB. Though the physical memory of each node is $132$ GB, we keep only $100$ GB as YARN resource allocated memory for each node. Therefore, the total physical memory for job execution is $100$ GB resulting $2$ executors per node and a total $6$ executors.  We reserve, $1$ core for operating system and hadoop daemons. Therefore, available total core is $11$. This leaves $5$ cores for each executor. We used these values of the run time resource parameters in all the experiments except the scalability test, where we have tested the approach with varied number of executors. 



\subsection{Comparison with state-of-the-art distributed systems}
In this section, we compare the performance of \textit{SPIN} with \textit{LU}. We report the running time of the competing approaches with increasing matrix dimension in Figure \ref{fig:Fastest-Running-Time}. We take the best wall clock time (fastest) among all the running time taken for different block sizes. It can be seen that, \textit{SPIN} takes the minimum amount of time for all matrix dimensions. Also, as expected the wall clock execution time increases with the matrix dimension, non-linearly (roughly as $O(n^{3})$). Also, the gap in wall clock execution time between both \textit{SPIN} and \textit{LU} increases monotonically with input matrix dimension. As we shall see in the next section, both \textit{LU} and \textit{SPIN} follow a \textit{U} shaped curve as a function of block sizes, hence allowing us to report the minimum wall clock execution time over all block sizes.

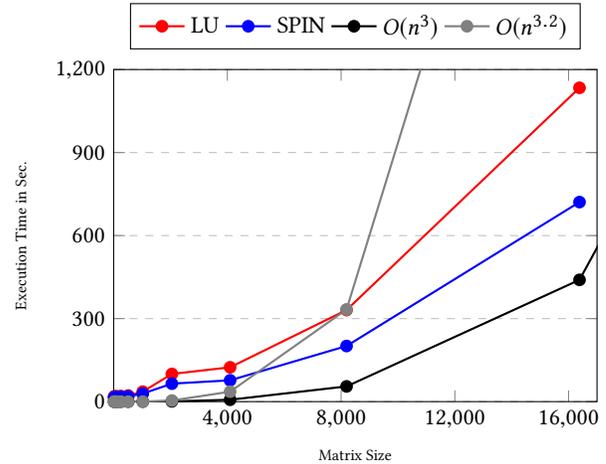
\begin{figure}
	\centering		
		\begin{tikzpicture}
		\begin{axis}[
        tick label style={/pgf/number format/fixed},
        scaled ticks=false,
		height=6cm,
		width=0.45\textwidth,
		xlabel={},
		ylabel={},
		xmin=0, xmax=17000,
		ymin=0, ymax=1200,
		xtick={4000,8000,12000,16000,20000},
		ytick={0,300,600,900,1200,1500,1800,2000},
		every axis plot/.append style={thick},
		xlabel={\scriptsize Matrix Size},
		ylabel={\scriptsize Execution Time in Sec.},
		legend style={at={(0.5,1.2)},
		anchor=north,legend columns=-1},
		ymajorgrids=true,
		grid style=dashed,
		]
		
		\addplot[
		color=red,
		mark=*,
		]
		coordinates {
			(16,18)
			(32,19)
			(64,19)
			(128,19)
			(256,20)
			(512,22)
			(1024,36)
			(2048,100)
			(4096,124)
			(8192,331)
			(16384,1134)
		};
		
		\addplot[
		color=blue,
		mark=*,
		]
		coordinates {
			(16,15)
			(32,16)
			(64,16)
			(128,16)
			(256,18)
			(512,19)
			(1024,29)
			(2048,65)
			(4096,77)
			(8192,201)
			(16384,721)
		};
        
        \addplot[
		color=black,
		mark=*,
		]
		coordinates {
			(16,0.0000004096)
			(32,0.0000032768)
			(64,0.0000262144)
			(128,0.0002097152)
			(256,0.0016777216)
			(512,0.013)
			(1024,0.107)
			(2048,0.859)
			(4096,6.8)
			(8192,55)
			(16384,440)
            (32768,3518)
		};
        
        \addplot[
		color=gray,
		mark=*,
		]
		coordinates {
			(16,0.00000047)
			(32,0.0000039)
			(64,0.000032)
			(128,0.00027)
			(256,0.0022)
			(512,0.018)
			(1024,0.15)
			(2048,3.95)
			(4096,36)
			(8192,333)
			(16384,3062)
            (32768,28147)
		};
	
		\legend{LU, SPIN, $O(n^{3})$, $O(n^{3.2})$}
		
		\end{axis}
		\end{tikzpicture}
	\caption{Fastest running time of LU and Strassen's based inversion among different block	sizes}
	\label{fig:Fastest-Running-Time}
\end{figure}

\subsection{Variation with partition size}
\label{sec:variation-with-partition-size}
In this experiment, we examine the performance of \textit{SPIN} with \textit{LU} with increasing partition size for each matrix size. We report the wall clock execution time of the approaches when partition size is increased within a particular matrix size. For each matrix size (from $(4096\times 4096)$ to $(16384\times 16384)$) we increase the partition size until we get a intuitive change in the results as shown in Figure. \ref{fig:Running-Time-Size}.

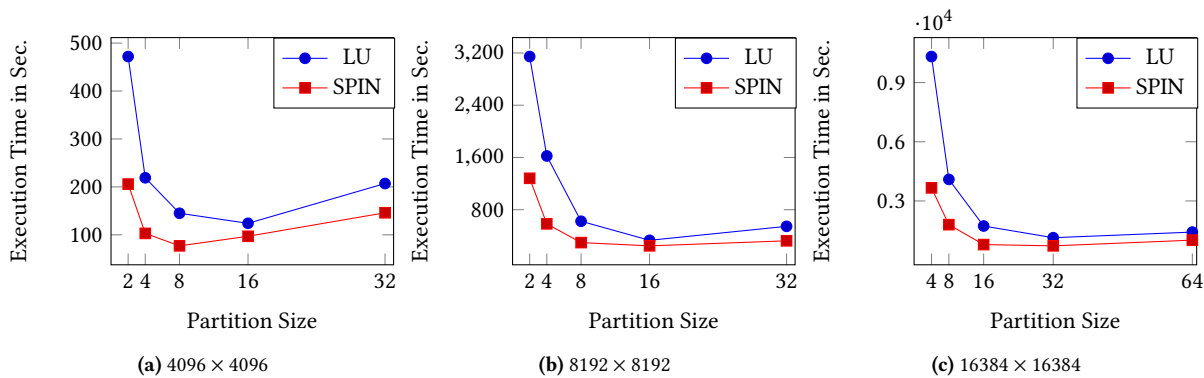
\begin{figure*}
	\begin{subfigure}[t]{.3\textwidth}
		\begin{tikzpicture}
		\begin{axis}[
		xmin=0,
		xmax=33,
		xtick={2,4,8,16,32},
		ytick={100,200,300,400,500},
		x tick label style={/pgf/number format/1000 sep=},
		xlabel={Partition Size},
		ylabel={Execution Time in Sec.},
		legend style={at={(0.79,1)},anchor=north},
		width=\textwidth,
		]
		\addplot coordinates {(2,472) (4,219) (8,145) (16,124) (32,207)};
		\addplot coordinates {(2,206) (4,103) (8,77) (16,97) (32,146)};
        \legend{LU, SPIN}
		\end{axis}
		\end{tikzpicture}
		\caption{$4096\times 4096$},
	\end{subfigure}%
	\begin{subfigure}[t]{.3\textwidth}
		\begin{tikzpicture}
		\begin{axis}[
		xmin=0,
		xmax=33,
		xtick={2,4,8,16,32},
		ytick={800,1600,2400,3200},
		x tick label style={/pgf/number format/1000 sep=},
		xlabel={Partition Size},
		ylabel={Execution Time in Sec.},
		legend style={at={(0.79,1)},anchor=north},
		width=\textwidth,
		]
		\addplot coordinates {(2,3146) (4,1623) (8,623) (16,331) (32,545)};
		\addplot coordinates {(2,1281) (4,583) (8,296) (16,246) (32,323)};
		\legend{LU, SPIN}
		\end{axis}
		\end{tikzpicture}
		\caption{$8192\times 8192$}
	\end{subfigure}%
	\begin{subfigure}[t]{.3\textwidth}
		\begin{tikzpicture}
		\begin{axis}[
		xmin=0,
		xmax=65,
		xtick={4,8,16,32,64},
		ytick={3000,6000,9000,12000},
		x tick label style={/pgf/number format/1000 sep=},
		xlabel={Partition Size},
		ylabel={Execution Time in Sec.},
		legend style={at={(0.79,1)},anchor=north},
		width=\textwidth,
		]
		\addplot coordinates {(4,10322) (8,4091) (16,1727) (32,1134) (64,1418)};	
		\addplot coordinates {(4,3660) (8,1796) (16,786) (32, 721) (64,1010)};
        \legend{LU, SPIN}
		\end{axis}
		\end{tikzpicture}
		\caption{$16384\times 16384$}
	\end{subfigure}%
	\caption{Comparing running time of LU and SPIN for matrix size $(4096\times 4096)$, $(8192\times 8192)$, $(16384\times 16384)$ for increasing partition size}\label{fig:Running-Time-Size}
\end{figure*}

It can be seen that both \textit{LU} and \textit{SPIN} follows a \textit{U} shape curve. However, \textit{SPIN} outperforms \textit{LU} when they have the same partition size, for all the matrix sizes. The reason of this is manifold. First of all, \textit{LU} requires $9$ times more $O\left(\frac{n}{b}\right)^{3}$ operations compared to a single operation of \textit{SPIN}. For small partition sizes, where \textit{leafNode} dominates the overall wall clock execution time, this cost is responsible for \textit{LU}'s slower performance.

Additionally, when the partition size increases, the number of recursion level also increases and consequently the cost of \textit{multiply} method increases which is the costliest method call. Though there is a difference between the number of recursion level for any partition size (), the additional matrix multiplication cost (as shown in Table \ref{tab:cost-LU-Spin}) provides enough cost to slowdown \textit{LU}'s performance.

\subsection{Comparison between theoretical and experimental result}
In this experiment, we compare the theoretical cost of \textit{SPIN} with the experimental wall clock execution time to validate our theoretical cost analysis. Figure \ref{fig:ThVsEx} shows the comparison for three matrix sizes (from $(4096\times 4096)$ to $(16384\times 16384)$ and for each matrix size with increasing partition size.

\begin{figure*}
	\begin{subfigure}[t]{.3\textwidth}
		\begin{tikzpicture}
		\begin{axis}[
		xmin=0,
		xmax=17,
		xtick={2,4,8,16},
		ytick={50,100,150,200},
		x tick label style={/pgf/number format/1000 sep=},
		xlabel={Partition Size},
		ylabel={Execution Time in Sec.},
		legend style={at={(0.5,1)},anchor=north},
		width=\textwidth,
		]
		\addplot coordinates {(2,197) (4,76) (8,46) (16,47)};     \addlegendentry{\small Theoretical}
		\addplot coordinates {(2,206) (4,103) (8,77) (16,97)};
   \addlegendentry{\small Experimental}
		\end{axis}
		\end{tikzpicture}
		\caption{$4096\times 4096$},
	\end{subfigure}%
	\begin{subfigure}[t]{.3\textwidth}
		\begin{tikzpicture}
		\begin{axis}[
		xmin=0,
		xmax=33,
		xtick={2,4,8,16,32},
		ytick={800,1600,2400,3200},
		x tick label style={/pgf/number format/1000 sep=},
		xlabel={Partition Size},
		ylabel={Execution Time in Sec.},
		legend style={at={(0.5,1)},anchor=north},
		width=\textwidth,
		]
		\addplot coordinates {(2,1559) (4,579) (8,337) (16,309) (32,465)};
    \addlegendentry{Theoretical}
		\addplot coordinates {(2,1281) (4,583) (8,296) (16,246) (32,323)};
\addlegendentry{Experimental}
		\end{axis}
		\end{tikzpicture}
		\caption{$8192\times 8192$}
	\end{subfigure}%
	\begin{subfigure}[t]{.3\textwidth}
		\begin{tikzpicture}
		\begin{axis}[
		xmin=0,
		xmax=65,
		xtick={4,8,16,32,64},
		ytick={1200,2400,3600,4800,6000},
		x tick label style={/pgf/number format/1000 sep=},
		xlabel={Partition Size},
		ylabel={Execution Time in Sec.},
		legend style={at={(0.5,1)},anchor=north},
		width=\textwidth,
		]
		\addplot coordinates {(4,4518) (8,2559) (16,2200) (32,2763) (64,5549)};	
    \addlegendentry{\small Theoretical}
		\addplot coordinates {(4,3660) (8,1796) (16,786) (32, 721) (64,2285)};
   \addlegendentry{\small Experimental}
		\end{axis}
		\end{tikzpicture}
		\caption{$16384\times 16384$}
	\end{subfigure}%
	\caption{Comparing theoretical and experimental running time of \textit{SPIN} for matrix size $(4096\times 4096)$, $(8192\times 8192)$, $(16384\times 16384)$ for increasing partition size}\label{fig:ThVsEx}
\end{figure*}
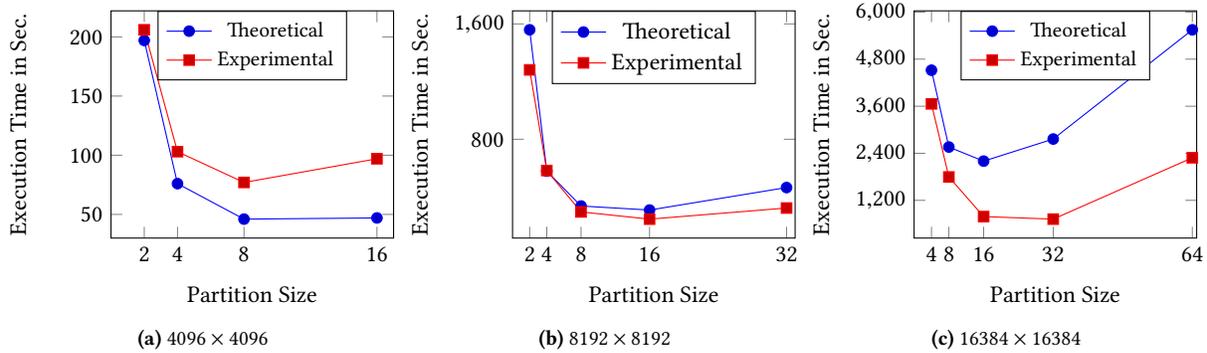

As expected, both theoretical and experimental wall clock execution time shows a \textit{U} shaped curve with increasing partition size. The reason is that, for smaller partition sizes, the block size becomes very large for large matrix size. As a result, the single node matrix inversion shares most of the execution time and subdues the effect of matrix multiplication execution time which are processed distributedly. That is why we find large execution time at beginning, which is also depicted in Table \ref{tab:wall-clock-methods}, where experimental wall clock execution time is tabulated for different methods used in the algorithm for matrix of dimension $4096$. It is seen that for $b=2$, the \textit{leafNode} cost is far more than matrix \textit{multiply} method.

Later, when partition size further increases, the leaf node cost drops sharply as the cost depends on $\frac{n^{3}}{b^{2}}$, which decreases the cost by square of partition size. On the other hand, the number of \textit{multiply} becomes large for enhanced recursion level, and thus the effective cost which subdues the effect of \textit{leafNode} cost. As in Table \ref{tab:wall-clock-methods}, for $b=8$ onwards the \textit{multiply} cost becomes more and more dominating resulting further increase in wall clock execution time.

\begin{table}
\caption{Experimental results of wall clock execution time of different methods in \textit{SPIN}(The unit of execution time is millisecond)}
\label{tab:wall-clock-methods}
\begin{minipage}{\columnwidth}
\begin{center}
\begin{tabular}{|c|c|c|c|c|}
\hline
Method & b = 2 & b = 4 & b = 8 & b = 16 \\
\hline
leafNode & 43504 & 11550 & 5040 & 3980 \\
\hline
breakMat & 178 & 441 & 901 & 1764 \\
\hline
xy & 2913 & 1353 & 693 & 309 \\
\hline
multiply & 7836 & 13116 & 23256 & 37968 \\
\hline
subtract & 1412 & 1854 & 2820 & 5592 \\
\hline
scalar & 333 & 728 & 1308 & 2450 \\
\hline
arrange & 307 & 685 & 1510 & 3074 \\
\hline
Total & 56483 & 29727 & 35528 &	55137 \\
\hline
\end{tabular}
\end{center}
\end{minipage}
\end{table}

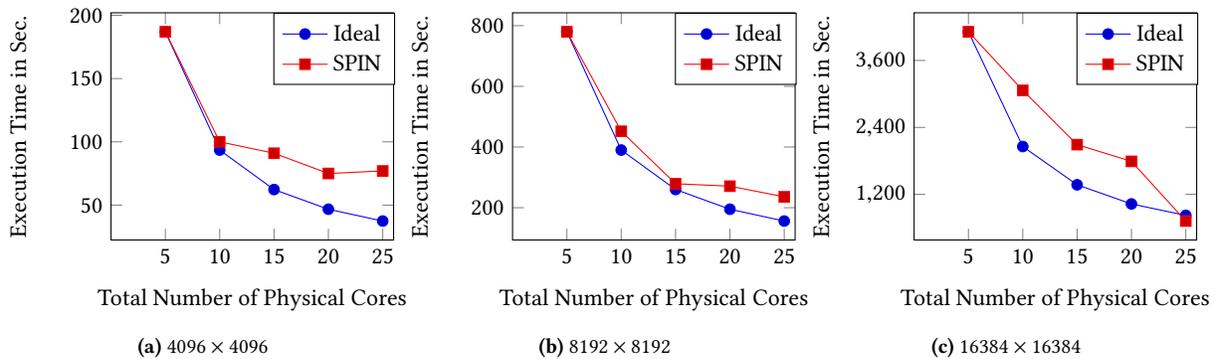
\begin{figure*}[t]
	\begin{subfigure}[t]{.3\textwidth}
		\begin{tikzpicture}
		\begin{axis}[
		xmin=0,
		xmax=26,
		xtick={5,10,15,20,25},
		ytick={50,100,150,200},
		x tick label style={/pgf/number format/1000 sep=},
		xlabel={Total Number of Physical Cores},
		ylabel={Execution Time in Sec.},
		legend style={at={(0.79,1)},anchor=north},
		width=\textwidth,
		]
		\addplot coordinates {(5,187) (10,93.5) (15,62.3) (20,46.75)(25,37.4)};
        \addlegendentry{Ideal}
		\addplot coordinates {(5,187) (10,100) (15,91) (20,75)(25,77)};        \addlegendentry{SPIN}
		\end{axis}
		\end{tikzpicture}
		\caption{$4096\times 4096$},
	\end{subfigure}%
	\begin{subfigure}[t]{.3\textwidth}
		\begin{tikzpicture}
		\begin{axis}[
		xmin=0,
		xmax=26,
		xtick={5,10,15,20,25},
		ytick={200,400,600,800},
		x tick label style={/pgf/number format/1000 sep=},
		xlabel={Total Number of Physical Cores},
		ylabel={Execution Time in Sec.},
		legend style={at={(0.79,1)},anchor=north},
		width=\textwidth,
		]
		\addplot coordinates {(5,780) (10,390) (15,260) (20,195) (25,156)};
   \addlegendentry{Ideal}
		\addplot coordinates {(5,780) (10,452) (15,279) (20,271) (25,236)};
\addlegendentry{SPIN}
		\end{axis}
		\end{tikzpicture}
		\caption{$8192\times 8192$}
	\end{subfigure}%
	\begin{subfigure}[t]{.3\textwidth}
		\begin{tikzpicture}
		\begin{axis}[
		xmin=0,
		xmax=26,
		xtick={5,10,15,20,25},
		ytick={1200,2400,3600,4800,6000},
		x tick label style={/pgf/number format/1000 sep=},
		xlabel={Total Number of Physical Cores},
		ylabel={Execution Time in Sec.},
		legend style={at={(0.79,1)},anchor=north},
		width=\textwidth,
		]
		\addplot coordinates {(5,4115) (10,2057) (15,1371) (20,1028) (25,823)};	
   \addlegendentry{Ideal}
		\addplot coordinates {(5,4115) (10,3064) (15,2091) (20,1791) (25,721)};
   \addlegendentry{SPIN}
		\end{axis}
		\end{tikzpicture}
		\caption{$16384\times 16384$}
	\end{subfigure}%
	\caption{The scalability of \textit{SPIN}, in comparison with ideal scalability (blue line), on matrix $(4096\times 4096)$, $(8192\times 8192)$ and $(16384\times 16384)$}\label{fig:scale}
\end{figure*}

\subsection{Scalability}
In this section, we investigate the scalability of \textit{SPIN}. For this, we generate three test cases, each containing a different set of two matrices of sizes equal to $(4096\times 4096)$, $(8192\times 8192)$ and $(16384\times 16384)$. The running time vs. the number of spark executors for these $3$ pairs of matrices is shown in Figure \ref{fig:scale}. The ideal scalability line (i.e. $T(n) = T(1)/n$ - where $n$ is the number of executors) has been over-plotted on this figure in order to demonstrate the scalability of our algorithm. We can see that \textit{SPIN} has a good scalability, with a minor deviation from ideal scalability when the size of the matrix is low (i.e. for $(4096\times 4096)$ and $(8192\times 8192)$).

\section{Conclusion}
In this paper, we have focused on the problem of distributed matrix inversion of large matrices using Spark framework. To make large scale matrix inversion faster, we have implemented Strassen's matrix inversion technique which requires six multiplications in each recursion step. We have given the detailed algorithm, called \textit{SPIN}, of the implementation and also presented the details of the cost analysis along with the baseline approach using LU decomposition. By doing that, we discovered that the primary bottleneck of inversion algorithm is matrix multiplications and that \textit{SPIN} is faster as it requires less number of multiplications compared to LU based approach. 

We have also performed extensive experiments on wall clock execution time of both the approaches for increasing partition size as well as increasing matrix size. Results showed that \textit{SPIN} outperformed LU for all the partition and matrix sizes and also the difference increases as we increase matrix size. We also showed the resemblance between theoretical and experimental findings of \textit{SPIN}, which validated our cost analysis. At last we showed that \textit{SPIN} has a good scalability with increasing matrix size.




\end{document}